\documentclass{article}

\PassOptionsToPackage{numbers, compress}{natbib}

\usepackage[preprint]{neurips_2024}

\usepackage[utf8]{inputenc} %
\usepackage[T1]{fontenc}    %
\usepackage{hyperref}       %
\usepackage{url}            %
\usepackage{booktabs}       %
\usepackage{amsfonts}       %
\usepackage{nicefrac}       %
\usepackage{microtype}      %
\usepackage{xcolor}         %

\usepackage{amsmath}
\usepackage{amssymb}
\usepackage{mathtools}
\usepackage{amsthm}

\usepackage[capitalize,noabbrev]{cleveref}

\usepackage{thm-restate}

\theoremstyle{plain}
\newtheorem{theorem}{Theorem}[section]

\newtheorem{lemma}[theorem]{Lemma}

\theoremstyle{definition}
\newtheorem{definition}[theorem]{Definition}

\theoremstyle{remark}

\usepackage[textsize=tiny]{todonotes}

\usepackage{algorithm}
\usepackage[noend]{algpseudocode}
\usepackage{wrapfig}

\usepackage{verbatim}

\usepackage{graphicx}
\usepackage{subcaption}
\usepackage{enumitem}

\usepackage{placeins} %

\DeclareMathOperator*{\argmin}{argmin}

\usepackage[textsize=tiny]{todonotes}

\newcommand{\eps}{\varepsilon}

\title{A Bi-metric Framework for Fast Similarity Search}

\author{%
  Haike Xu \\
  MIT \\
  \texttt{haikexu@mit.edu} \\
  \And
  Sandeep Silwal \\
  MIT \\
  \texttt{silwal@mit.edu} \\
  \And
  Piotr Indyk \\
  MIT \\
  \texttt{indyk@mit.edu} \\
}

\begin{document}

\maketitle

\begin{abstract}
We propose a new ``bi-metric'' framework for designing nearest neighbor data structures. Our framework assumes two dissimilarity functions: a {\em ground-truth} metric that is accurate but expensive to compute, and a {\em proxy} metric that is cheaper but less accurate. In both theory and practice, we show how to construct data structures using only the proxy metric such that the query procedure achieves the accuracy of the expensive metric, while only using a limited number of calls to both metrics.  Our theoretical results instantiate this framework for two popular nearest neighbor search algorithms: DiskANN and Cover Tree. In both cases we show that, as long as the proxy metric used to construct the data structure approximates the ground-truth metric up to a bounded factor, our data structure achieves arbitrarily good approximation guarantees with respect to the ground-truth metric. On the empirical side, we apply the framework to the text retrieval problem with two dissimilarity functions evaluated by ML models with vastly different computational costs. We observe that for almost all data sets in the MTEB benchmark, our approach achieves a considerably better accuracy-efficiency tradeoff than the alternatives, such as re-ranking\footnote{Our code is publicly available at \url{https://github.com/xuhaike/Bi-metric-search}}.
\end{abstract}

\section{Introduction}

Similarity search is a versatile and popular approach to data retrieval. It assumes that the data items of interest (text passages, images, etc.) are equipped with a distance function, which for any pair of items estimates their similarity or dissimilarity\footnote{To simplify the presentation, throughout this paper we assume a dissimilarity function.}. Then, given a “query” item, the goal is to return the data item that is most similar to the query.  From the algorithmic perspective, this approach is formalized as the nearest neighbor search (NN) problem: given a set of $n$ points $P$ in a metric space $(X,D)$, build a data structure that, given any query point $q \in X$, returns $p \in P$ that minimizes  $D(p,q)$ . 
In many cases, the items are represented by high-dimensional feature vectors and $D$ is induced by the Euclidean distance between the vectors. In other cases, $D(p,q)$ is computed by a dedicated procedure given $p$ and $q$ (e.g., by a cross-encoder).

Over the last decade, mapping the data items to feature vectors, or estimation of similarity between pairs of data items, is often done using machine learning models. (In the context of text retrieval, the first task is achieved by constructing bi-encoders \cite{karpukhin-etal-2020-dense, neelakantan2022text, gao2021simcse, wang2024text}, while the second task uses cross-encoders \cite{gao2021rethink, nogueira2020document, nogueira2020passage}). This creates efficiency bottlenecks, as high-accuracy machine learning models are often large and costly to use, while cheaper models do not achieve the state-of-the-art accuracy. Furthermore, the high-accuracy models are often proprietary, accessible only through a limited interface, and frequently updated without notice. This motivates the study of “the best of both worlds” solutions which utilize many types of models to achieve favorable tradeoffs between efficiency, accuracy and flexibility.

One popular method for combining multiple models is based on “re-ranking” \cite{liu2009learning}. It assumes two models: one model evaluating the metric $D$, which has high accuracy but is less efficient; and another model computing a ``proxy’’ metric $d$, which is cheap but less accurate. The algorithm uses the second model ($d$) to retrieve a large (say, $k=1000$) number of data items with the highest similarity to the query, and then uses the first model ($D$) to select the most similar items.  The hyperparameter $k$ controls the tradeoff between the accuracy and efficiency.  To improve the efficiency further, the retrieval of the top-$k$ items is typically accomplished using approximate nearest neighbor data structures. Such data structures are constructed for the proxy metric $d$, so they remain stable even if the high-accuracy metric $D$ undergoes frequent updates. 

Despite its popularity, the re-ranking approach suffers from several issues:

\begin{enumerate}[leftmargin=*]
\item The overall accuracy is limited by the accuracy of the cheaper model. To illustrate this phenomenon, suppose that $D$ defines the “true” distance, while $d$ only provides a “$C$-approximate” distance, i.e., that the values of $d$ and $D$ for the same pairs of items differ by at most a factor of $C>1$. Then the re-ranking approach can only guarantee that the top reported item is a $C$-approximation, namely that its distance to the query is at most $C$ times the distance from the query to its true nearest neighbor according to $D$. This occurs because the first stage of the process, using the proxy $d$, might not retain the most relevant items.
\item Since the set of the top-$k$ items with respect to the more accurate model depends on the query, one needs to perform at least a linear scan over all $k$ data items retrieved using the proxy metric $d$. This computational cost can be reduced by decreasing $k$, but at the price of reducing the accuracy.
\end{enumerate}

\paragraph{Our results} In this paper we show that, in both theory and practice, it is possible to combine cheap and expensive models to achieve approximate nearest neighbor data structures that inherit the accuracy of expensive models while significantly reducing the overall computational cost. Specifically, we propose a {\em bi-metric} framework for designing nearest neighbor data structures with the following properties:

\begin{itemize}[leftmargin=*]
\item The algorithm for creating the data structure uses only the proxy metric $d$, making it efficient to construct,
\item The algorithm for answering the nearest neighbor query leverages both models, but performs only a sub-linear number of evaluations of $d$ and $D$, 
\item The data structure achieves the accuracy of the expensive model.
\end{itemize}

Our framework is quite general, and is applicable to any approximate nearest neighbor data structure that works for general metrics. Our {\em theoretical} study analyzes this approach when applied to two popular algorithms: DiskANN~\cite{jayaram2019diskann} and Cover Tree~\cite{beygelzimer2006cover}. Specifically, we show the following theorem statement. We use $\lambda_d$ to refer to the doubling dimension with respect to metric $d$ (a measure of intrinsic dimensionality, see Definition \ref{def:doubling_dim}).

\begin{theorem}[Summary, see Theorems \ref{thm:diskann_main} and \ref{thm:ct_main}]
Given a dataset $X$ of $n$ points, $\mathtt{Alg} \in \{\text{DiskANN}, \text{Cover Tree}\}$, and a fixed metric $d$, let $S_{\mathtt{Alg}}(n, \eps, \lambda_d)$ and $Q_{\mathtt{Alg}}(\eps, \lambda_d)$ denote the space and query complexity respectively of the standard datastructure for $\mathtt{Alg}$ which reports a $1+\eps$ nearest neighbor in $X$ for any query (all for a fixed metric $d$).

Consider two metrics $d$ and $D$ satisfying Equation \ref{eq:relationship_metric}. Then for any $\mathtt{Alg} \in \{\text{DiskANN}, \text{Cover Tree}\}$, we can build a corresponding datastructure $\mathcal{D}_{\mathtt{Alg}}$ on $X$ with the following properties:
\begin{enumerate}
    \item When construction $\mathcal{D}_{\mathtt{Alg}}$, we only access metric $d$,
    \item The space used by $\mathcal{D}_{\mathtt{Alg}}$ can be bounded by $\tilde{O}(S_{\mathtt{Alg}}(n, \eps/C, \lambda_d))$\footnote{$\tilde{O}$ hides logarithm dependencies in the aspect ratio.},
    \item Given any query $q$,  $\mathcal{D}_{\mathtt{Alg}}$ invokes $D$ at most $\tilde{O}(Q_{\mathtt{Alg}}(\eps/C, \lambda_d))$ times,
    \item $\mathcal{D}_{\mathtt{Alg}}$ returns a $1+\eps$ approximate nearest neighbor of $q$ in $X$ under metric $D$.
\end{enumerate}
\end{theorem}

We emphasize that the data structures we study achieve an arbitrarily small approximation factor of $1+\eps$, even though $d$ approximates $D$ only up to some constant factor $C$. The proof of the theorem crucially uses the properties of the underlying data structures. It is an interesting open direction to determine if our bi-metric framework can be theoretically instantiated for other popular nearest neighbor algorithms, such as those based on locality sensitive hashing.

To demonstrate the {\em practical} applicability of the bi-metric  framework, we apply it to the text retrieval problem. Here, the data items are text passages, and the goal is to retrieve a passage from a large collection that is most relevant to a query passage. We instantiated our framework with the DiskANN algorithm, using a high-quality ``SFR-Embedding-Mistral'' model \cite{SFRAIResearch2024} to define $D$, and a lower-quality ``bge-micro-v2'' model \cite{bgemicro} to define $d$. Both metrics $d(p,q)$ and $D(p,q)$ are induced by the Euclidean distance between the embeddings of $p$ and $q$ using the respective models. The sizes of the two models differ by 3 orders of magnitude, making $D$ much more expensive to evaluate than $d$. We evaluated the retrieval quality of our approach on a benchmark collection of 15 MTEB retrieval data sets \cite{thakur2021beir}, comparing it to the re-ranking approach, which retrieves the closest data items to the query with respect to $d$ and re-ranks using $D$. We observe that for almost all data sets, our approach achieves a considerably better accuracy-efficiency tradeoff than re-ranking. In particular, for several data sets, such at HotpotQA \cite{yang-etal-2018-hotpotqa}, our approach achieves state-of-the-art retrieval accuracy using up to $\textbf{4x}$ fewer evaluations of the expensive model.

\subsection{Related Work}

\paragraph{Graph-based algorithms for similarity search}
The algorithms studied in this paper rely on graph-based data structures for (approximate) nearest neighbor search. Such data structures work for general metrics, which, during the pre-processing, are approximated by carefully constructed graphs.  Given the graph and the query point, the query answering procedure greedily searches the graph to identify the nearest neighbors. Graph-based algorithms have been extensively studied both in theory~\cite{KrauthgamerL04, beygelzimer2006cover} and in practice~\cite{fu2019fast, jayaram2019diskann,malkov2018efficient,harwood2016fanng}. See~\cite{clarkson2006nearest,wang2021comprehensive} for an overview of these lines of research. 

\paragraph{Data analysis with proxy metrics} We are aware of the following prior papers that design and analyze algorithms that use proxy and ground-truth metrics~\cite{GuhaGRS15, moseley2021hierarchical, silwal2023kwikbucks, batenietal23}. These papers focus on clustering, while in this work we focus on the nearest neighbor search. The paper~\cite{moseley2021hierarchical} is closest to our work, as it uses approximate nearest neighbor as a subroutine when computing the clustering. However, their algorithm only achieves the (lower) accuracy of the cheaper model, while our algorithms retains the (higher) accuracy of the expensive one.

More broadly, the idea of using cheap but noisy filters to identify candidate items that are then re-ranked using expensive and accurate methods is a popular approach in many applications, including recommendation systems~\cite{liu2022neural} and computer vision~\cite{zhong2017re}.

\section{Preliminaries}

\paragraph{Problem formulation:}

We assume that we are give two metrics over the pairs of points from $X$:
\begin{itemize}[leftmargin=*]
\item The {\em ground truth} metric $D$, which for any pair $p,q \in X$ returns the ``true'' dissimilarity between $p$ and $q$, and
\item The {\em proxy} metric $d$, which provides a cheap approximation to the ground truth metric.
\end{itemize}
Throughout the paper, we think of $D$ as `expensive' to evaluate, and $d$ as the cheaper, but noisy, proxy.

We first consider the {\em exact} nearest neighbor search. Here, the goal is to build an index structure that, given a query $q\in X$, returns $p^*\in P$ such that $p^* = \mbox{arg min}_{p \in P} D(q,p)$. 
The algorithm for constructing the data structure has access to the proxy metric $d$, but not the ground truth metric $D$. The algorithm for answering a query $q$ has access to both metrics. The complexity of the query-answering procedure is measured by counting the number of evaluations of the ground truth metric $D$.

As described in the introduction, the above formulation is motivated by the following considerations:
\begin{itemize}[leftmargin=*]
\item Evaluating ground truth metric $D$ is typically very expensive. 
Therefore, our cost model for the query answering procedure only accounts for the number of such evaluations. 
\item The metric $D$ can change in unpredictable ways after the index is constructed, e.g., due to model update. Hence we do not assume that the algorithm can access the (current version of) $D$ during the index construction phase.
\end{itemize}

Our formulation can be naturally extended to more general settings, such as:
\begin{itemize}[leftmargin=*]
\item $(1+\eps)$-approximate nearest neighbor search, where the goal is to find any $p^*$ such that \[ D(q,p^*) \le (1+\eps) \min_{p \in P} D(q,p).\]
\item $k$-nearest neighbor search, where the goal is to find the set of $k$ nearest neighbors of $q$ in $P$ with respect to $D$. If the algorithm returns a set $S'$ of $k$ points  that is different than the set $S$ of true $k$ nearest neighbor, the quality of the answer is measured by computing the Recall rate or NDCG score \cite{ndcg}. 

\end{itemize}

\paragraph{Assumptions about metrics:} Clearly, if the metrics $d$ and $D$ are not related to each other, the data structure constructed using $d$ alone does not help with the query retrieval. Therefore, we assume that the two metrics are related through the following definition.

\begin{definition}\label{def:distortion}
Given a set of $n$ points $P$ in a metric space $X$ and its distance function $D$, we say the distance function $d$ is a $C$-approximation of $D$ if for all $x,y \in X$, 
\begin{equation}\label{eq:relationship_metric}
    d(x,y) \le D(x,y) \le C \cdot d(x,y).
\end{equation}
\end{definition}
For a fixed metric $d$ and any point $p\in X$, radius $r>0$, we use $B(p,r)$ to denote the ball with radius $r$ centered at $p$, i.e. $B(p,r)=\{q\in X: d(p,q)\le r\}$. In our paper, the notion of \emph{doubling-dimension} is central. It is a measure of intrinsic dimensionality of datasets which is popular in analyzing high dimensional datasets, especially in the context of nearest neighbor search algorithms \cite{GuptaKL03, KrauthgamerL04, beygelzimer2006cover, IndykN07, Har-PeledK13, NarayananSIZ21, NEURIPS2023_worstcase}.

\begin{definition}[Doubling Dimension]\label{def:doubling_dim}
 We say a point set $X$ has doubling dimension $\lambda_d$ with respect to metric $d$ if for any $p\in X$, and radius $r>0$, $X\cap B(p,2r)$ can be covered by at most $2^{\lambda_d}$ balls with radius $r$.   
\end{definition}
Finally, for a metric $d$, $\Delta_d$ is the aspect ratio of the input set $X$, i.e., the ratio between the diameter and the distance of the closest pair.

\section{Theoretical analysis}\label{sec:theory}
We instantiate our \emph{bi-metric} framework for two popular nearest neighbor search algorithms: DiskANN and Cover Tree. The goal of our bi-metric framework is to first create a data structure using the proxy (cheap) metric $d$, but solve nearest neighbor to $1+\eps$ accuracy for the expensive metric $D$. Furthermore, the query step should invoke the metric $D$ judiciously, as the number of calls to $D$ is the measure of efficiency. Our theoretical query answering algorithms do not use calls to $d$ at all.

We note that, if we treat the proxy data structure as a {\em black box}, we can only guarantee that it returns a $C$-approximate nearest neighbor with respect to $D$. Our theoretical analysis overcomes this, and shows that calling $D$ a sublinear number of times in the query phase (for DiskANN and Cover Tree) allows us to obtain {\em arbitrarily accurate} neighbors for $D$.

At a high level, the unifying theme of the algorithms that we analyze (DiskANN and Cover Tree) is that they both crucially use the concept of a \emph{net}: given a parameter $r$, a $r$-net is a small subset of the dataset guaranteeing that every data point is within distance $r$ to the subset in the net. Both algorithms (implicitly or explicitly), construct nets of various scales $r$ which help route queries to their nearest neighbors in the dataset. The key insight is that a net of scale $r$ for metric $d$ is also a net under metric $D$, but with the larger scale $Cr$. Thus, if we construct smaller nets for metric $d$, they can also function as nets for the expensive metric $D$ (which we don't access during our data structure construction). Care must be taken to formalize this intuition and we present the details below. 

We remark that the intuition we gave clearly does not generalize for nearest neighbor algorithms which are fundamentally different, such as locality sensitive hashing. For such algorithms, it is not clear if any semblance of a bi-metric framework is algorithmically possible, and this is an interesting open direction. In the main body, we present the (simpler) analysis of DiskANN and defer the analysis of Cover Tree to Appendix \ref{sec:covertree}.

\subsection{DiskANN}\label{sec:diskann_analysis}
First, some helpful background on the algorithm is given.

\subsubsection{Preliminaries for DiskANN}\label{diskann_prelim}
In Section \ref{diskann_prelim}, we only deal with a single metric $d$.
We first need the notion of an $\alpha$-shortcut reachability graph. Intuitively, it is an unweighted graph $G$ where the vertices correspond to points of a dataset $X$ such that nearby points (geometrically) are close in graph distance. 

\begin{definition}[$\alpha$-shortcut reachability \cite{NEURIPS2023_worstcase}]\label{def:alpha_graph}
Let $\alpha \ge 1$. We say a graph $G=(X,E)$ is $\alpha$-shortcut reachable from a vertex $p$ under a given metric $d$ if for any other vertex $q$, either $(p,q)\in E$, or there exists $p'$ s.t. $(p,p')\in E$ and $d(p',q)\cdot\alpha\le d(p,q)$. We say a graph $G$ is $\alpha$-shortcut reachable under metric $d$ if $G$ is $\alpha$-shortcut reachable from any vertex $v\in X$.
\end{definition}

The main analysis of \cite{NEURIPS2023_worstcase} shows that (the `slow preprocessing version' of ) DiskANN outputs an $\alpha$-shortcut reachability graph.

\begin{theorem}[\cite{NEURIPS2023_worstcase}]\label{thm:alpha_graph} Given a dataset $X$, $\alpha \ge 1$, and fixed metric $d$ the slow preprocessing DiskANN algorithm (Algorithm $4$ in \cite{NEURIPS2023_worstcase}) outputs a $\alpha$-shortcut reachibility graph $G$ on $X$ as defined in Definition \ref{def:alpha_graph} (under metric $d$). The space complexity of $G$ is $n \cdot \alpha^{O(\lambda_d)} \log(\Delta_d)$.
\end{theorem}

Given an $\alpha$-reachability graph on a dataset $X$ and a query point $q$, \cite{NEURIPS2023_worstcase} additionally show that the greedy search procedure of Algorithm \ref{alg:search-algorithm_diskann} finds an accurate nearest neighbor of $q$ in $X$.

\begin{theorem}[Theorem $3.4$ in \cite{NEURIPS2023_worstcase}]\label{thm:greedy_analysis}
For $\eps \in (0, 1)$, there exists an $\Omega(1/\eps)$-shortcut reachable graph index for a metric $d$ with max degree $\text{Deg} \le (1/\eps)^{O(\lambda_d)}\log(\Delta_d)$ (via Theorem \ref{thm:alpha_graph}). For any query $q$, Algorithm \ref{alg:search-algorithm_diskann} on this graph index finds a $(1 + \eps)$ nearest neighbor of $q$ in $X$ (under metric $d$) in $S\le O(\log(\Delta_d))$ steps and makes at most $S \cdot \text{Deg} \le (1/\eps)^{O(\lambda_d)}\log(\Delta_d)^2$ calls to $d$.
\end{theorem}

\paragraph{The main theorem.}
We are now ready to state the main theorem of Section \ref{sec:diskann_analysis}.

\begin{theorem}\label{thm:diskann_main}
Let $Q_{\mathtt{DiskAnn}}(\eps, \Delta_d, \lambda_d) = (1/\eps)^{O(\lambda_d)}\log(\Delta_d)^2 $ denote the query complexity of the standard DiskANN data structure\footnote{I.e., the upper bound on the number of calls made to $d$ on any query}, where we build and search using the same metric $d$. Now consider two metrics $d$ and $D$ satisfying Equation \ref{eq:relationship_metric}. Suppose we build an $C/\eps$-shortcut reachability graph $G$ using Theorem \ref{thm:alpha_graph} for metric $d$, but search using metric $D$ in Algorithm \ref{alg:search-algorithm_diskann} for a query $q$. Then the following holds:
\begin{enumerate}
    \item The space used by $G$ is  at most $n \cdot (C/\eps)^{O( \lambda_d)} \log(\Delta_d)$.
    \item Running Algorithm \ref{alg:search-algorithm_diskann} using $D$ finds a $1+\eps$ approximate nearest neighbor of $q$ in the dataset $X$ (under metric $D$).
    \item On any query $q$, Algorithm \ref{alg:search-algorithm_diskann} invokes $D$ at most $Q_{\mathtt{DiskAnn}}(\eps/C, C \Delta_d, \lambda_d)$.
    \end{enumerate}
\end{theorem}

To prove the theorem, we first show that a shortcut reachability graph of $d$ is also a shortcut reachability graph of $D$, albeit with slightly different parameters.

\begin{lemma}\label{lm:reachable}
Suppose metrics $d$ and $D$ satisfy relation \eqref{eq:relationship_metric}. Suppose $G=(X,E)$ is $\alpha$-shortcut reachable under $d$ for $\alpha>C$. Then $G=(X,E)$ is an $\alpha/C$-shortcut reachable under $D$.
\end{lemma}
\begin{proof}
    Let $(p,q)$ be a pair of distinct vertices such that $(p,q) \not \in E$. Then we know that there exists a $(p,p') \in E$ such that $d(p',q) \cdot \alpha \le d(p,q)$. From relation \eqref{eq:relationship_metric}, we have $\frac{1}C \cdot D(p', q) \cdot \alpha \le  d(p', q) \cdot \alpha \le d(p,q) \le D(p,q)$,
    as desired.    
\end{proof}

\begin{proof}[Proof of Theorem \ref{thm:diskann_main}]
    By Lemma~\ref{lm:reachable}, the graph $G=(X,E)$ constructed for metric $d$ is also a $O(1/\eps)$ reachable for the other metric $D$. Then we simply invoke Theorem \ref{thm:greedy_analysis} 
    for a $(1/\eps)$-reachable graph index for metric $D$ with degree limit $Deg\le (C/\eps)^{O(\lambda_d)}\log(\Delta_d)$ and the number of greedy search steps $S\le O(\log(C\Delta_d))$. Thus the total number of $D$ distance call bounded by $(C/\eps)^{O(\lambda_d)}\log(C\Delta_d)^2\le Q_{\mathtt{DiskAnn}}(\eps/C, C \Delta_d, \lambda_d)$.
    This proves the accuracy bound as well as the number of calls we make to metric $D$ during the greedy search procedure of Algorithm \ref{alg:search-algorithm_diskann}. The space bound follows from Theorem \ref{thm:alpha_graph}, since $G$ is a $C/\eps$-reachability graph for metric $d$.
\end{proof}

\section{Experiments}\label{sec:experiments}
In this section we present an experimental evaluation of our approach. The starting point of our implementation is the DiskANN based algorithm from Theorem~\ref{thm:diskann_main}, which we engineer to optimize the performance\footnote{Our experiments are run on 56 AMD EPYC-Rome processors with 400GB of memory and 4 NVIDIA RTX 6000 GPUs. Our experiment in Figure~\ref{fig:main-table} takes roughly 3 days.}. We compare it to two other methods on all 15 MTEB retrieval tasks \cite{thakur2021beir}. 

\subsection{Experiment Setup}

\paragraph{Methods}
We evaluate the following methods. $\mathcal{Q}$ denotes the query budget, i.e., the maximum number of calls an algorithm can make to $D$ during a query. We vary  $\mathcal{Q}$ in our experiments.
 
\begin{itemize}[leftmargin=*]
\item \underline{Bi-metric (our method)}: We build a graph index with the cheap distance function $d$ (we discuss our choice of graph indices in the experiments shortly). Given a query $q$, we first search for $q$'s top-$\mathcal{Q}/2$ nearest neighbor under metric $d$. Then, we start a second-stage search from the $\mathcal{Q}/2$ returned vertices using distance function $D$ on the same graph index until we reach the quota $\mathcal{Q}$. We report the 10 closest neighbors seen so far by distance function $D$.

\item \underline{Bi-metric (baseline)}: This is the standard retrieve + rerank method that is widely popular. We build a graph index with the cheap distance function $d$. Given a query $q$, we first search for $q$'s top-$\mathcal{Q}$ nearest neighbor under metric $d$. As explained below, we can assume that empirically the first step returns the \emph{true} top-$\mathcal{Q}$ nearest neighbors under $d$. Then, we calculate  distance using $D$ for all the $\mathcal{Q}$ returned vertices and report the top-10. 

\item \underline{Single metric}: This is the standard nearest neighbor search with $D$ distance. We build the graph index directly with the expensive distance function $D$. Given a query $q$, we do a standard greedy search to get the top-10 closest vertices to $q$ with respect to distance $D$ until we reach quota $\mathcal{Q}$. We help this method and ignore the large number of $D$ distance calls in the indexing phase and only count towards the quota in the search phase. Note that this method doesn't satisfy our ``bi-metric'' formulation as it uses an extensive number of $D$ distance calls ($\Omega(n)$ calls) in index construction. However, we implement it for comparison since it represents a natural baseline, if one does not care about the prohibitively large number of calls made to $D$ during index building.
\end{itemize}

For both Bi-metric methods (ours and baseline), in the first-stage search under distance $d$, we initialize the parameters of the graph index so that empirically, it returns the true nearest neighbors under distance $d$. This is done by setting the `query length' parameter $L$ to be $30000$ for dataset with corpus size $>10^6$ (Climate-FEVER~\cite{diggelmann2020climatefever}, FEVER~\cite{thorne-etal-2018-fever}, HotpotQA~\cite{yang-etal-2018-hotpotqa}, MSMARCO~\cite{bajaj2018msmarco}, NQ~\cite{nq}, DBPedia~\cite{dbpedia}) and 5000 for the other datasets. Our choice of $L$ is large enough to ensure that the returned vertices are almost true nearest neighbors under distance $d$. For example, the standard parameters used are a factor of $10$ smaller. We also empirically verified that the nearest neighbors returned for $d$ with such large values of $L$ corroborated with published MTEB benchmark values~\footnote{from \url{https://huggingface.co/spaces/mteb/leaderboard}} of the distance proxy model. Thus for both Bi-metric methods (our method and baseline), we can equivalently think of the first-stage as running a brute-force method for $d$.

\paragraph{Datasets} We experiment with all of the following 15 MTEB retrieval datasets: Arguana \cite{wachsmuth2018arguana}, ClimateFEVER\cite{diggelmann2020climatefever}, CQADupstackRetrieval\cite{hoogeveen2015cqadupstack}, DBPedia\cite{dbpedia}, FEVER\cite{thorne-etal-2018-fever}, FiQA2018\cite{fiqa}, HotpotQA\cite{yang-etal-2018-hotpotqa}, MSMARCO\cite{bajaj2018msmarco}, NFCorpus\cite{nfcorpus}, NQ\cite{nq}, QuoraRetrieval\cite{thakur2021beir} 
SCIDOCS\cite{scidocs}, SciFact\cite{scifact}, Touche2020\cite{touche}, TRECCOVID\cite{treccovid}. As a standard practice, we report the results on these dataests' test split, except for MSMARCO where we report the results on its dev split.

\paragraph{Embedding Models} We select the current top-ranked model ``SFR-Embedding-Mistral'' as our expensive model to provide groundtruth metric $D$. Meanwhile, we select three models on the pareto curve of the MTEB retrieval size-average score plot to test how our method performs under different model scale combinations. These three small models are ``bge-micro-v2'', ``gte-small'', ``bge-base-en-v1.5''. Please refer to Table~\ref{tab:models} for details.

\paragraph{Nearest Neighbor Search Algorithms} The nearest neighbor search algorithms we employ in our experiments are DiskANN\cite{jayaram2019diskann} and NSG\cite{NSG}. 
The parameter choices for DiskANN are $\alpha=1.2$, $l\_build=125$, $max\_outdegree=64$ (the standard choices used in ANN benchmarks \cite{annbenchmark}). The parameter choices for NSG are the same as the authors' choices for GIST1M dataset \cite{PQ}: 
$K=400$, $L=400$, $iter=12$, $S=15$, $R=100$. 
NSG also requires building a knn-graph with efanna, where we use the standard parameters: 
$L=60$, $R=70$, $C=500$.

\paragraph{Metric} Given a fixed expensive distance function quota $\mathcal{Q}$, we report the accuracy of retrieved results for different methods. We always insure that all algorithms never use more than $\mathcal{Q}$ expensive distance computations. Following the MTEB retrieval benchmark, we report the NDCG@10 score. Following the standard nearest neighbor search algorithm benchmark metric, we also report the Recall@10 score compared to the true nearest neighbor according to the expensive metric $D$.

\subsection{Experiment Results and Analysis}
Please refer to Figure~\ref{fig:main-table} 
for our results with $d$ distance function set to ``bge-micro-v2'' and $D$ distance function set to ``SFR-Embedding-Mistral'', with the underlying graph index set to the DiskANN algorithm. To better focus on the convergence speed of different methods, we cut off the y-axis at a relatively high accuracy, so some curves may not start from x equals 0 if their accuracy doesn't reach the threshold.

We can see that our method converges to the optimal accuracy much faster than Bi-metric (baseline) and Single metric in most cases. For example for HotpotQA, the NDCG@10 score achieved by the baselines for $8000$ calls to $D$ is comparable to our method, using less than $2000$ calls to $D$, leading to $\textbf{>4x}$ fewer evaluations of the expensive model.

This means that utilizing the graph index built for the distance function proxy to perform a greedy search using $D$ is more efficient than naively iterating the returned vertex list to re-rank using $D$ (baseline). It is also noteworthy to see that our method converges faster than ``Single metric'' in all the datasets except FiQA2018 and TRECCOVID, especially in the earlier stages. This phenomenon happens even if ``Single metric'' is allowed infinite expensive distance function calls in its indexing phase to build the ground truth graph index. This suggests that the quality of the underlying graph index is not as important, and the early routing steps in the searching algorithm can be guided with a cheap distance proxy functions to save expensive distance function calls.

Similar conclusion holds for the recall plot (see Figure~\ref{fig:main-table-recall}) as well, where our method has an even larger advantage over Bi-metric (baseline) and is also better than the Single metric in most cases, except for FEVER, FiQA2018, and HotpotQA. We report the results of using different model pairs, using the NSG algorithm as our graph index, and measuring Recall@10 in Appendix~\ref{sec:complete-experiment-results}. Please see their ablation studies in Section~\ref{sec:ablation}.

\begin{figure}[!ht]
\centering
\includegraphics[width=0.99\textwidth]{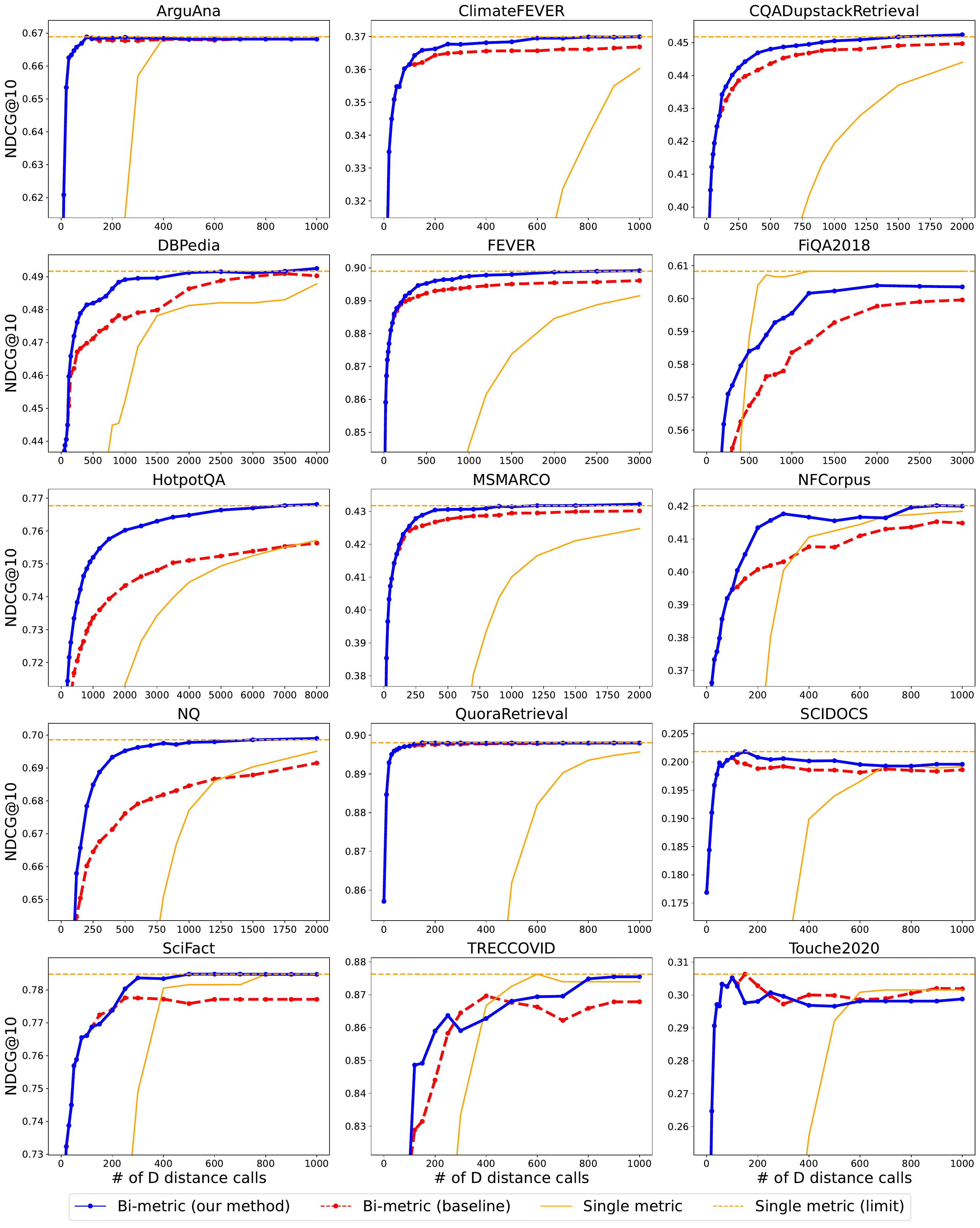}   
\caption{Results for 15 MTEB Retrieval datasets. The x-axis is the number of expensive distance function calls. The y-axis is the NDCG@10 score. The cheap model is ``bge-micro-v2'', the expensive model is ``SFR-Embedding-Mistral'', and the nearest neighbor search algorithm used is DiskANN.}
\label{fig:main-table}
\end{figure}

\subsection{Ablation studies}\label{sec:ablation}
We investigate the impact of different components of our method. All ablation studies are run on HotpotQA dataset as it is one of the largest and most difficult retrieval dataset where the performance gaps between different methods are substantial.

\paragraph{Different model pairs} Fixing the expensive model to be ``SFR-Embedding-mistral'' \cite{SFRAIResearch2024}, we experiment with 2 other cheap models from the model list of the MTEB retrieval benchmark: ``gte-small'' \cite{li2023towardsgte} and ``bge-base'' \cite{bgeembedding}. These models have different sizes and capabilities, summarized in Table~\ref{tab:models}. For complete results on all 15 MTEB Retrieval datasets for different cheap models, we refer to Figures~\ref{fig:diskann-gte-small-ndcg},~\ref{fig:diskann-gte-small-recall},~\ref{fig:diskann-bge-base-ndcg}, and~\ref{fig:diskann-bge-base-recall} in Appendix~\ref{sec:complete-experiment-results}. Here, we only focus on the HotpotQA dataset.

\begin{table}[!ht]
\begin{tabular}{cccc}
\toprule
Model Name            & Embedding Dimension & Model Size & MTEB Retrieval Score \\
\midrule
SFR-Embedding-Mistral \cite{SFRAIResearch2024} & 4096                & 7111M      & 59                   \\
bge-base-en-v1.5 \cite{bgeembedding}      & 768                 & 109M       & 53.25                \\
gte-small \cite{li2023towardsgte}  & 384                 & 33M        & 49.46                \\
bge-micro-v2 \cite{bgemicro}         & 384                 & 17M        & 42.56                \\
\bottomrule
\end{tabular}
\caption{Different cheap models used in our experiments}
\label{tab:models}
\end{table}

From Figure~\ref{fig:hotpotqa-multi}, we can observe that the improvement of our method is most substantial when there is a large gap between the qualities of the cheap and expensive models. This is not surprising: If the cheap model has already provided enough accurate distances, simple re-ranking can easily get to the optimal retrieval results with only a few expensive distance calls. Note that even in the latter case, our second-stage search method still performs at least as good as re-ranking. Therefore, we believe that the ideal scenario for our method is a small and efficient model deployed locally, paired with a remote large model accessed online through API calls to maximize the advantages of our method.

\paragraph{Different nearest neighbor search algorithms} We implement our method with another popular empirical nearest neighbor search algorithm called NSG \cite{fu2019fast}. 
Because the authors' implementation of NSG only supports $l_2$ distances, we first normalize all the embeddings and search via $l_2$ distance. This may cause some performance drops. Therefore, we are not comparing the results between the DiskANN and NSG algorithms, but only results from different methods, fixing the graph index.

In Figure~\ref{fig:nsg-bge-micro-ndcg} and ~\ref{fig:nsg-bge-micro-recall} in the appendix, we observe that our method still performs the best compared to Bi-metric (baseline) and single metric in most cases, demonstrating that our bi-metric framework can be applied to other graph-based nearest neighbor search algorithms as well.

\begin{figure}[ht]
    \centering
    \begin{minipage}{0.49\textwidth}
        \centering
        \includegraphics[width=0.99\linewidth]{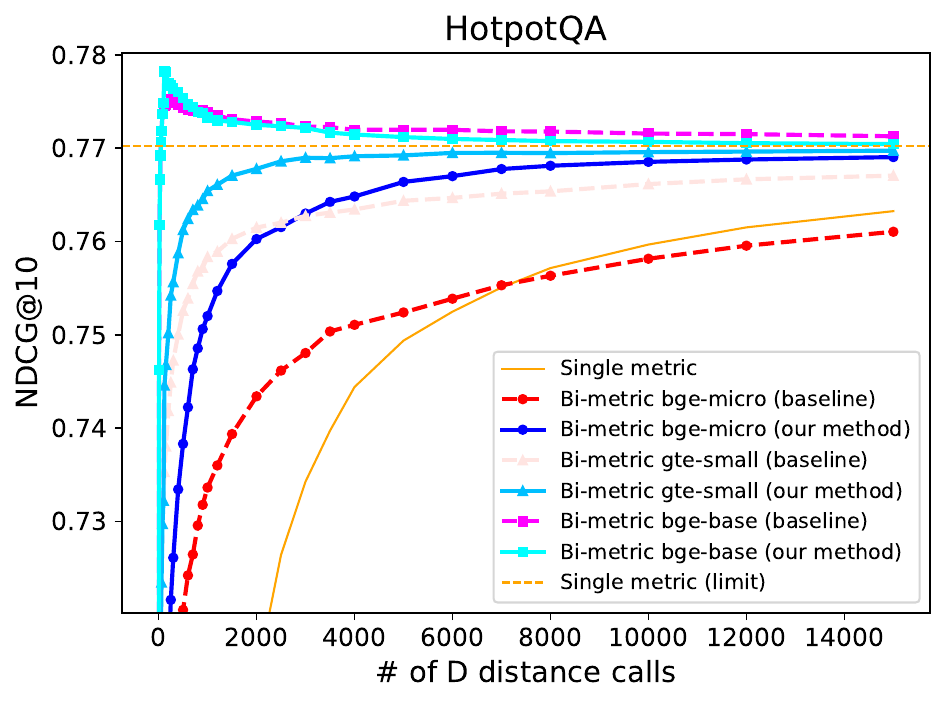}
        \caption{HotpotQA test results for different models as the distance proxy. Blue / skyblue / cyan curves represent Bi-metric (our method) with bge-micro / gte-small / bge-base models. Red / rose / magenta curves represent Bi-metric (baseline) with bge-micro / gte-small / bge-base models}
        \label{fig:hotpotqa-multi}
    \end{minipage}
    \hfill
    \begin{minipage}{0.49\textwidth}
        \centering
        \includegraphics[width=0.99\linewidth]{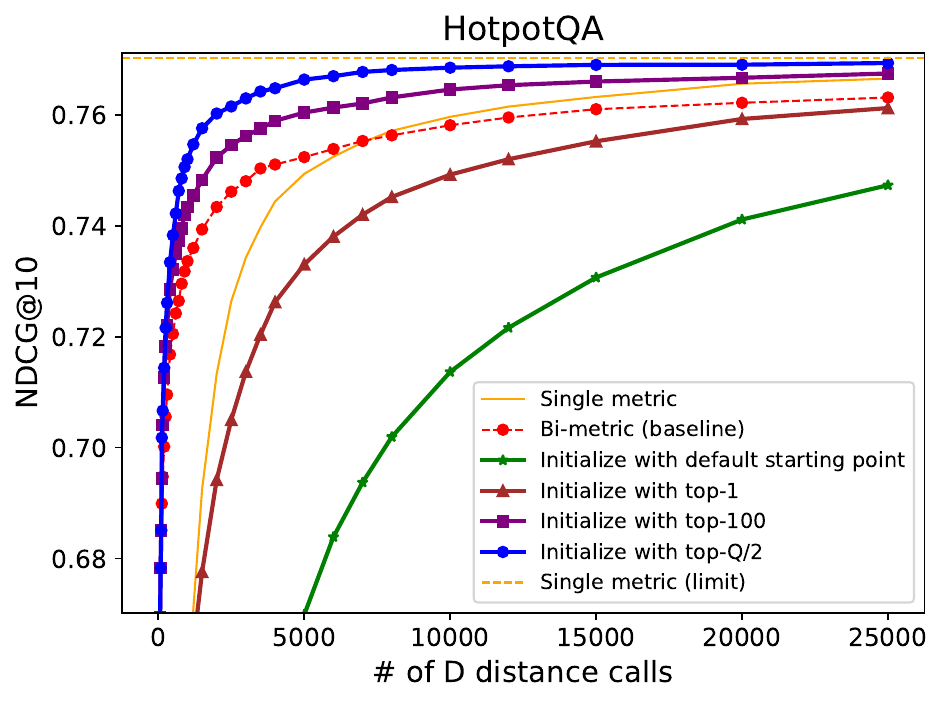}
        \caption{HotpotQA test results for different search initializations for the second-stage search of Bi-metric (our method). Blue / purple / brown / green curves represent initializing our second-stage search with top-$\mathcal{Q}/2$, top-100, top-1, or the default vertex.}
        \label{fig:hotpotqa-start}
    \end{minipage}
\end{figure}

\vspace{-0.2cm}

\paragraph{Impact of the first stage search} In the second-stage search of our method, we start from multiple points returned by the first-stage search via the cheap distance metric. We investigate how varying the starting points for the second-stage search impact the final results. We try four different
setups:

\begin{itemize}[leftmargin=*]
    \item Default: We start a standard nearest neighbor search using metric $D$ from the default entry point of the graph index, which means that we don't use the first stage search.
    \item Top-$K$ points retrieved by the first stage search: Suppose our expensive distance calls quota is $\mathcal{Q}$. We start our second search from the top $K$ points retrieved by the first stage search. We experiment with the following different choices of $K$: $K_1=1$, $K_{100}=100$, $K_{\mathcal{Q}/2}=\max(100,\mathcal{Q}/2)$ (note $K_{\mathcal{Q}/2}$ is the choice we use in Figure~\ref{fig:main-table}).
\end{itemize}

From Figure~\ref{fig:hotpotqa-start}, we observe that utilizing results from the first-stage search helps the second-stage search to find the nearest neighbor quicker. 

For comparison, we experiment with initializing the second-stage search from the default starting point (green), which means that we don't need the first-stage search and only use the graph index built from $d$ (cheap distance function). The DiskANN algorithm still manages to improve as the allowed number of $D$ distance calls increases, but it converges the slowest compared to all the other methods.

Using multiple starting points further speeds up the second stage search. If we only start with the top-1 point from the first stage search (brown), its NDCG curve is still worse than Bi-metric (baseline, red) and Single metric (orange). As we switch to top-100 (purple) or top-$\mathcal{Q}/2$ (blue) starting points, the NDCG curves increase evidently.

We provide two intuitive explanations for these phenomena. First, the approximation error of the cheap distance function doesn't matter that much in the earlier stage of the search, so the first stage search with the cheap distance function can quickly get to the true `local' neighborhood without any expensive distance calls, thus saving resource for the second stage search. Second, the ranking provided by the cheap distance function is not accurate because of its approximation error, so starting from multiple points should give better results than solely starting from the top few, which also justifies the advantage of our second-stage search over re-ranking.

\bibliographystyle{plain}
\bibliography{ref}

\appendix

\section{Query algorithm of DiskANN}

\begin{algorithm}[!ht]
\caption{\label{alg:search-algorithm_diskann} DiskANN-GreedySearch($q, d$)}
\begin{algorithmic}[1]
\State \textbf{Input}: Graph index $G=(X,E)$, distance function $d$, starting point $s$, query point $q$
    
\State \textbf{Output}: visited vertex list $U$
\State $s \gets$ an arbitrary starting point in $X$
\State $A\gets \{s\}$
\State $U\gets \varnothing$

\While{$A\setminus U \neq\varnothing$}
    \State $v\gets\argmin_{v\in A\setminus U} d(x_v,q)$
    \State $A \gets A\cup Neighbors(v)$ \Comment{Neighbors in $G$}
    \State $U \gets U\cup v$
    \If{$|A|> 1$}
        \State $A\gets \text{ closest vertex to $q$ in $A$}$
    \EndIf
\EndWhile
\State sort $U$ in increasing distance from $q$
\State \textbf{return} $U$

\end{algorithmic}
\end{algorithm}

\section{Analysis of Cover Tree}\label{sec:covertree}
We now analyze Cover Tree under the bi-metric framework. First, some helpful background is presented below. 

\subsubsection{Preliminaries for Cover Tree}\label{covertree_prelim}
The notion of a cover is central.
\begin{definition}[Cover]\label{def:cover}
A $r$-cover $\mathcal{C}$ of a set $X$ given a metric $d$ is defined as follows. Initially $\mathcal{C} = \emptyset$. Run the following two steps until $X$ is empty.
\begin{enumerate}
    \item Pick an arbitrary point $x \in X$ and remove $B(x, r) \cap X$ from $X$.
    \item Add $x$ to $\mathcal{C}$.
\end{enumerate}

\end{definition}

Note that a cover with radius $r$ satisfies the following two properties: every point in $X$ is within distance $r$ to some point in $\mathcal{C}$ (under the same metric $d'$), and all points in $\mathcal{C}$ are at least distance $r$ apart from each other. 

We now introduce the cover tree datastructure of \cite{beygelzimer2006cover}. For the data structure, we create a sequence of covers $\mathcal{C}_{-1}, \mathcal{C}_0, \ldots$. Every $\mathcal{C}_i$ is a layer in the final Cover Tree $\mathcal{T}$.

\begin{algorithm}[H]
\caption{\label{alg:cover_tree_create}Cover Tree Data structure}
\begin{algorithmic}[1]
\State \textbf{Input:} A set $X$ of $n$ points, metric $d$, real number $T \ge 1$.
\State \textbf{Output:} A tree on $X$
\Procedure{Cover-Tree}{$d, T$}
\State WLOG, all distances between points in $X$ under $d$ are in $(1,\Delta]$ by scaling.
\State $\mathcal{C}_{-1}=\mathcal{C}_0= X$
\State $\mathcal{C}_i$ is a $2^i/T$-cover of $\mathcal{C}_{i-1}$ for $i>0$ under metric $d$
\State $\mathcal{C}_i \subseteq \mathcal{C}_{i-1}$ for $i>0$.
\State $t = O(\log(\Delta T))$ \Comment{$t$ is the number of levels of $\mathcal{T}$}
\For{$i = -1$ to $t$}
\State $\mathcal{C}_i$ corresponds to tree nodes of $\mathcal{T}$ on level $i$
\State Each $p \in \mathcal{C}_{i-1} \setminus \mathcal{C}_i$ is connected to exactly one $p \in \mathcal{C}_i$ such that $d(p, p') \le 2^i/T$
\EndFor
\State \textbf{Return} tree $\mathcal{T}$
\EndProcedure
\end{algorithmic}
\end{algorithm}

\begin{lemma}[Theorem 1 in \cite{beygelzimer2006cover}]\label{lem:space}
$\mathcal{T}$ takes $O(n)$ space, regardless of the value of $T$.
\end{lemma}
\begin{proof}
We use the \emph{explicit} representation of $\mathcal{T}$ (as done in \cite{beygelzimer2006cover}), where we coalesce all nodes
in which the only child is a self-child. Thus, every node either has a parent other than the self-parent or a child other than the self-child. This gives an $O(n)$ space bound, independent of all other parameters.
\end{proof}

We note that it is possible to construct the cover tree data structure of Algorithm \ref{alg:cover_tree_create} in time $2^{O(\lambda_d)} n \log n$, but it is not important to our discussion \cite{beygelzimer2006cover}.

Now we describe the query procedure. Here, we can query with a metric $D$ that is possibly different than the metric $d$ used to create $\mathcal{T}$ in Algorithm \ref{alg:cover_tree_create}.

\begin{algorithm}[H]
\caption{\label{alg:cover_tree_search}Cover Tree Search}
\begin{algorithmic}[1]
\State \textbf{Input:} Cover tree $\mathcal{T}$ associated with point set $X$ , query point $q$, metric $D$, accuracy $\eps \in (0, 1)$.
\State \textbf{Output:} A point $p \in X$
\Procedure{Cover-Tree-Search}{}
\State $t \gets$ number of levels of $\mathcal{T}$
\State  $Q_t \gets \mathcal{C}_t$ \Comment{We use the covers that define $\mathcal{T}$}
\State $i\gets t$
\While{$i \ne -1$}
\State $Q=\{p \in \mathcal{C}_{i-1} : p\text{ has a parent in }Q_i\}$
\State $Q_{i-1} = \{ p \in Q : D(q,p) \le D(q, Q) + 2^i \}$
\If{$D(q, Q_{i-1}) \ge 2^i(1+1/\eps)$ }
\State Exit the while loop.
\EndIf
\State $i \gets i-1$
\EndWhile
\State \textbf{Return} point $p \in Q_{i-1}$ that is closest to $q$ under $D$
\EndProcedure
\end{algorithmic}
\end{algorithm}

\subsubsection{The main theorem}
We construct a cover tree $\mathcal{T}$ using metric $d$ and $T$ from Equation \ref{eq:relationship_metric} in Algorithm \ref{alg:cover_tree_create}. Upon a query $q$, we search for an approximate nearest neighbor in $\mathcal{T}$ in Algorithm \ref{alg:cover_tree_search}, using metric $D$ instead. Our main theorem is the following. 

\begin{theorem}\label{thm:ct_main}
Let $Q_{\mathtt{Cover Tree}}(\eps, \Delta_d, \lambda_d) = 2^{O(\lambda_d)} \log(\Delta_d) + (1/\eps)^{O(\lambda_d)}$ denote the query complexity of the standard cover tree datastructure, where we set $T = 1$ in Algorithm \ref{alg:cover_tree_create} and build and search using the same metric $d$. Now consider two metrics $d$ and $D$ satisfying Equation \ref{eq:relationship_metric}. Suppose we build a cover tree $\mathcal{T}$ with metric $d$ by setting $T = C$ in Algorithm \ref{alg:cover_tree_create}, but search using metric $D$ in Algorithm \ref{alg:cover_tree_search}. Then the following holds:
\begin{enumerate}
    \item The space used by $\mathcal{T}$ is $O(n)$.
    \item Running Algorithm \ref{alg:cover_tree_search} using $D$ finds a $1+\eps$ approximate nearest neighbor of $q$ in the dataset $X$ (under metric $D$).
    \item On any query, Algorithm \ref{alg:cover_tree_search} invokes $D$ at most 
    \[C^{O(\lambda_d)}\log(\Delta_d) + (C/\eps)^{O(\lambda_d)} = \tilde{O}(Q_{\mathtt{Cover Tree}}(\Omega(\eps/C), \Delta_d, \lambda_d)).\] times.
    \end{enumerate}
\end{theorem}

Two prove Theorem \ref{thm:ct_main}, we need to: (a) argue correctness and (b) bound the number of times Algorithm \ref{alg:cover_tree_search} calls its input metric $D$. While both follow from similar analysis as in \cite{beygelzimer2006cover}, it is not in a black-box manner since the metric we used to search $\mathcal{T}$ in Algorithm \ref{alg:cover_tree_search} is different than the metric used to build $\mathcal{T}$ in Algorithm \ref{alg:cover_tree_create}.

We begin with a helpful lemma.

\begin{lemma}\label{lem:descend}
  For any $p \in \mathcal{C}_{i-1}$, the distance between $p$ and any of its descendants in $\mathcal{T}$ is bounded by $2^i$ \emph{under} $D$.
\end{lemma}
\begin{proof}
    The proof of the lemma follows from Theorem 2 in \cite{beygelzimer2006cover}. There, it is shown that for any $p \in \mathcal{C}_{i-1}$ the distance between $p$ and any descendant $p'$  is bounded by $d(p, p' ) \le \sum_{j = -\infty}^{i-1} 2^j/T = 2^i/T$, implying the lemma after we scale by $C$ due to Equation \ref{eq:relationship_metric} (note we set $T = C$ in the construction of $\mathcal{T}$ in Theorem \ref{thm:ct_main}).
\end{proof}

We now argue accuracy.

\begin{theorem}\label{thm:ct_acc}
Algorithm \ref{alg:cover_tree_search} returns a $1+\eps$-approximate nearest neighbor to query $q$ under $D$.
\end{theorem}
\begin{proof}
Let $p^*$ be the true nearest neighbor of query $q$. Consider the leaf to root path starting from $p^*$. We first claim that if $Q_i$ contains an ancestor of $p^*$, then $Q_{i-1}$ also contains an ancestor $q_{i-1}$ of $p^*$. To show this, note that $D(p^*, q_{i-1}) \le 2^i$ by Lemma \ref{lem:descend}, so we always have 
\[D(q, q_{i-1}) \le D(q, p^*) + D(p^*, q_{i-1}) \le D(q, Q) + 2^i,\] meaning $q_{i-1}$ is included in $Q_{i-1}$.

When we terminate, either we end on a single node, in which case we return $p^*$ exactly (from the above argument), or when $D(q, Q_{i-1}) \ge 2^{i}(1 + 1/\eps)$. In this latter case, we additionally know that 
\[D(q, Q_{i-1}) \le D(q, p^*) + D(p^*, Q_{i-1}) \le D(q, p^*) + 2^i\] since an ancestor of $p^*$ is contained in $Q_{i-1}$ (namely $q_{i-1}$ from above). But the exit condition implies 
\[ 2^i(1 + 1/\eps) \le D(q, p^*) + 2^i \implies 2^i \le \eps D(q, p^*),  \]
which means 
\[D(q, Q_{i-1}) \le  D(q, p^*) + 2^i \le D(q, p^*) + \eps D(q, p^*) = (1+\eps) D(q, p^*),\]
as desired.
\end{proof}

Finally, we bound the query complexity. The following follows from the arguments in \cite{beygelzimer2006cover}.
\begin{theorem}\label{thm:ct_query}
The number of calls to $D$ in Algorithm \ref{alg:cover_tree_search} is bounded by $C^{O( \lambda_{d})} \cdot \log(\Delta_d C) + (C/\eps)^{O(\lambda_{d})}$.
\end{theorem}
\begin{proof}[Proof Sketch]
The bound follows from \cite{beygelzimer2006cover} but we briefly outline it here. The query complexity is dominated by the size of the sets $Q_{i-1}$ in Line $9$ as the algorithm proceeds. We give two ways to bound $Q_{i-1}$. Before that, note that the points $p$ that make up $Q_{i-1}$ are in a cover (under $d$) by the construction of $\mathcal{T}$, so they are all separated by distance at least $\Omega(2^i/C)$ (under $d$). Let $p^*$ be the closest point to $q$ in $X$.

\begin{itemize}[leftmargin=*]
\item  \textbf{Bound 1}: In the iterations where $D(q, p^*) \le O(2^i)$, we have the diameter of $Q_{i-1}$ under $D$ is at most $O(2^i)$ as well. This is because an ancestor $q_{i-1} \in C_{i-1}$ of $p^*$ is in $Q$ of line $8$ (see proof of Theorem \ref{thm:ct_acc}), meaning $D(q, Q) \le O(2^i)$ due to Lemma \ref{lem:descend}. Thus, any point $p \in Q_{i-1}$ satisfies $D(q, p) \le D(q, Q) + 2^i = O(2^i)$. From Equation \ref{eq:relationship_metric}, it follows that the diameter of $Q_{i-1}$ under $d$ is also at most $O(2^i)$. We know the points in $Q_{i-1}$ are separated by mutual distance at least $\Omega(2^i/C)$ under $d$, implying that  $|Q_{i-1}| \le C^{O( \lambda_{d})}$ in this case by a standard packing argument. This case can occur at most $O(\log(\Delta C))$ times, since that is the number of different levels of $\mathcal{T}$.
    \item \textbf{Bound 2}:  Now consider the case where $D(q,p^*) \ge \Omega(2^i)$. In this case, we have that the points in $Q_{i-1}$ have diameter at most $O(2^{i}/\eps)$ from $q$ (under $D$), due to the condition of line $10$. Thus, the diameter is also bounded by $O(2^{i}/\eps)$  under $d$.  By a standard packing argument, this means that $|Q_{i-1}| \le (C/\eps)^{O(\lambda_d)}$, since again $Q_{i-1}$ are mutually separated by distance at least $\Omega(2^i/C)$ under $d$. However, our goal is to show that the number of iterations where this bound is relevant is at most $O(\log(1/\eps))$. Indeed, we have $D(q,Q_{i-1}) \le O(2^i/\eps)$, meaning $2^i \ge \Omega(\eps D(q, Q_{i-1})) \ge \Omega(\eps D(q, p^*))$  Since we are decrementing the index $i$ and are in the case where $D(q,p^*) \ge \Omega(2^i)$, this can only happen for $O(\log(1/\eps))$ different $i$'s.
\end{itemize}

Combining the two bounds proves the theorem. \end{proof}

The proof of Theorem \ref{thm:ct_main} follows from combining Lemmas \ref{lem:space} and Theorems \ref{thm:ct_acc} and \ref{thm:ct_query}.

\section{Complete experimental results}\label{sec:complete-experiment-results}
We report the empirical results of using different embedding models as distance proxy, using the NSG algorithm, and measuring Recall@10.

\begin{enumerate}
    \item We report the results of using ``bge-micro-v2'' as the distance proxy $d$ and using DiskANN for building the graph index. See Figure~\ref{fig:main-table-recall} for Recall@10 metric plots.
    \item We report the results of using ``gte-small'' as the distance proxy $d$ and using DiskANN for building the graph index. See Figure~\ref{fig:diskann-gte-small-ndcg} for NDCG@10 metric plots and Figure~\ref{fig:diskann-gte-small-recall} for Recall@10 metric plots.
    \item We report the results of using ``bge-base-en-v1,5'' as the distance proxy $d$ and using DiskANN for building the graph index. See Figure~\ref{fig:diskann-bge-base-ndcg} for NDCG@10 metric plots and Figure~\ref{fig:diskann-bge-base-recall} for Recall@10 metric plots.
    \item We report the results of using ``bge-micro-v2" as the distance proxy $d$ and using NSG for building the graph index. See Figures~\ref{fig:nsg-bge-micro-ndcg} for NDCG@10 metric plots and ~\ref{fig:nsg-bge-micro-recall} for Recall@10 metric plots.
\end{enumerate}

We can see that for all the different cheap distance proxies (``bge-micro-v2'' \cite{bgeembedding}, ``gte-small'' \cite{li2023towardsgte}, ``bge-base-en-v1.5'' \cite{bgeembedding}) and both nearest neighbor search algorithms (DiskANN \cite{jayaram2019diskann} and NSG \cite{fu2019fast}), our method has better NDCG and Recall results on most datasets. Moreover, naturally the advantage of our method over Bi-metric (baseline) is larger when there is a large gap between the qualities of the cheap distance proxy $d$ and the ground truth distance metric $D$. This makes sense because as their qualities converge, the cheap proxy alone is enough to retrieve the closest points to a query for the expensive metric $D$.

\begin{figure}[!h]
\centering
\includegraphics[width=0.99\textwidth]{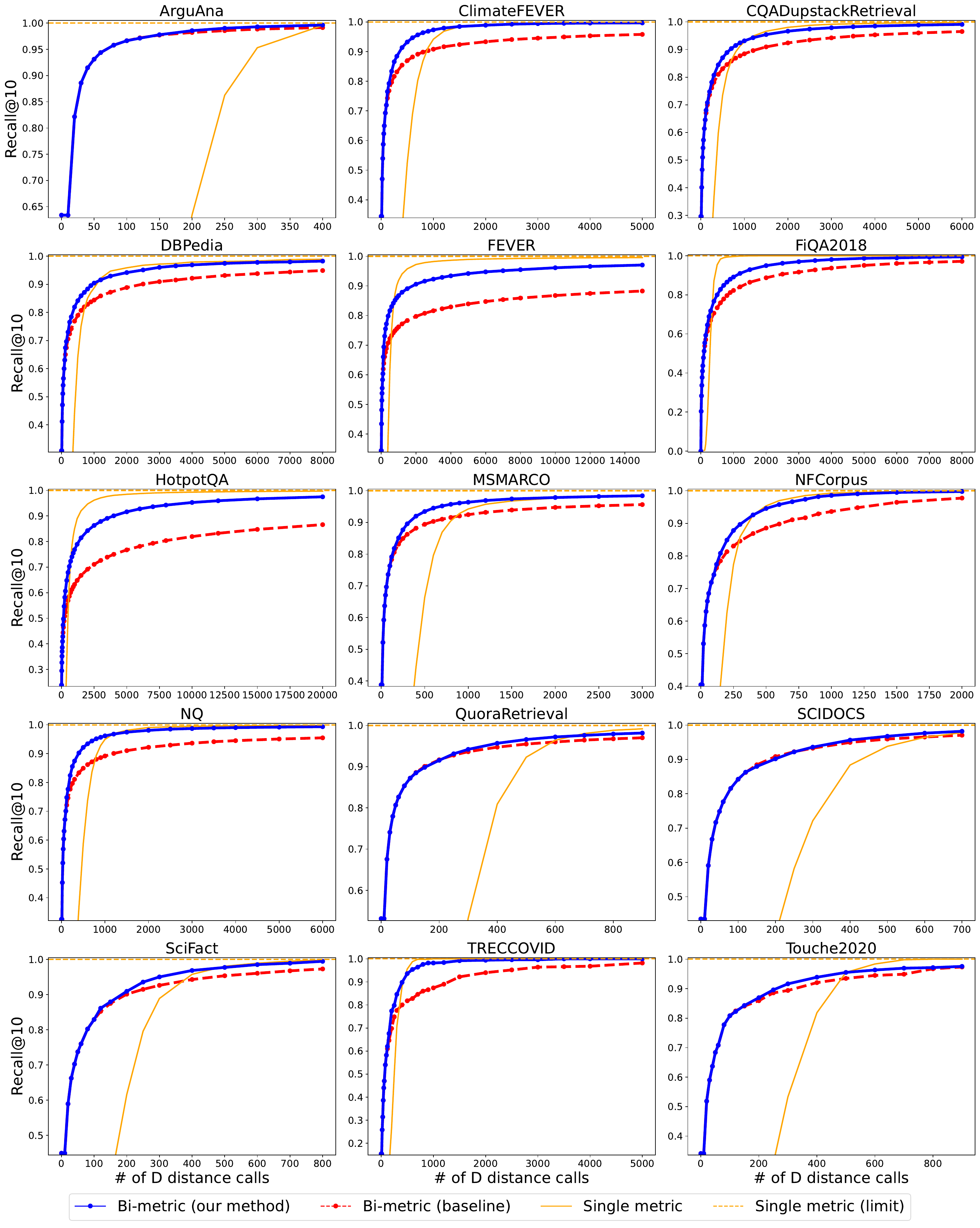}
\caption{Results for 15 MTEB Retrieval datasets. The x-axis is the number of expensive distance function calls. The y-axis is the Recall@10 score. The cheap model is ``bge-micro-v2'', the expensive model is ``SFR-Embedding-Mistral'', and the nearest neighbor search algorithm used is DiskANN.}
\label{fig:main-table-recall}
\end{figure}

\begin{figure}[!h]
\centering
\includegraphics[width=0.99\textwidth]{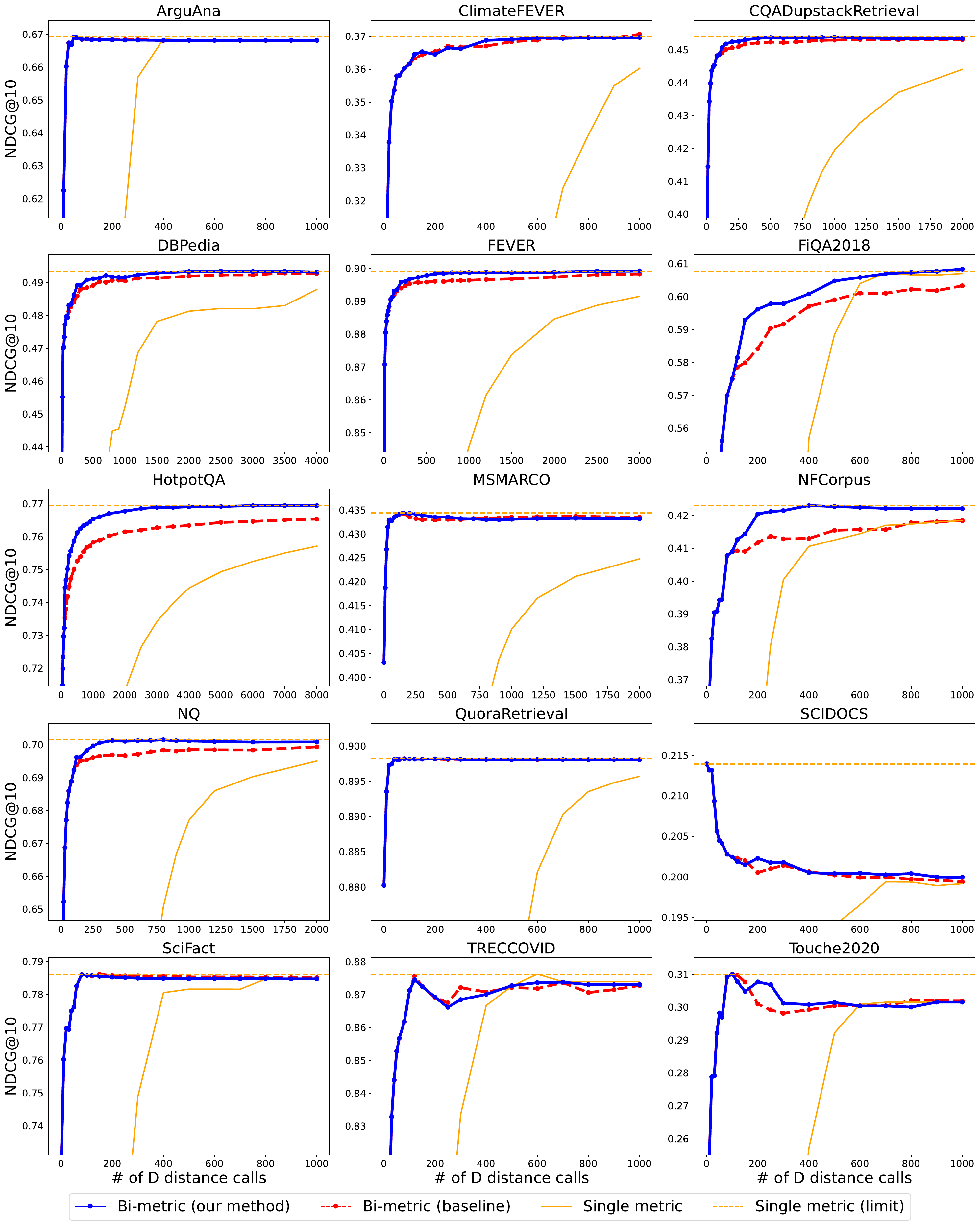}
\caption{Results for 15 MTEB Retrieval datasets. The x-axis is the number of expensive distance function calls. The y-axis is the NDCG@10 score. The cheap model is ``gte-small'', the expensive model is ``SFR-Embedding-Mistral'', and the nearest neighbor search algorithm used is DiskANN.}
\label{fig:diskann-gte-small-ndcg}
\end{figure}

\begin{figure}[!h]
\centering
\includegraphics[width=0.99\textwidth]{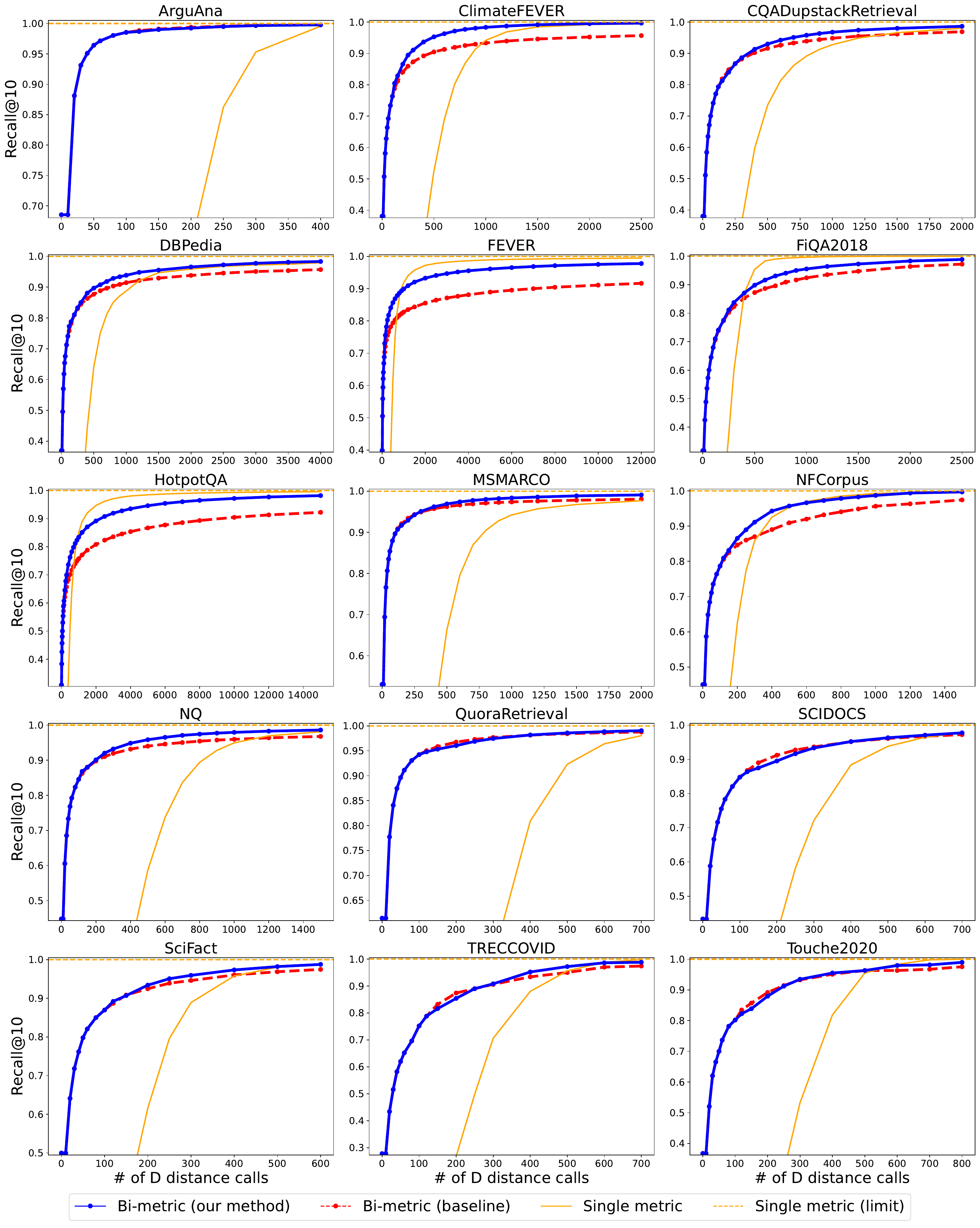}
\caption{Results for 15 MTEB Retrieval datasets. The x-axis is the number of expensive distance function calls. The y-axis is the Recall@10 score. The cheap model is ``gte-small'', the expensive model is ``SFR-Embedding-Mistral'', and the nearest neighbor search algorithm used is DiskANN.}
\label{fig:diskann-gte-small-recall}
\end{figure}

\begin{figure}[!h]
\centering
\includegraphics[width=0.99\textwidth]{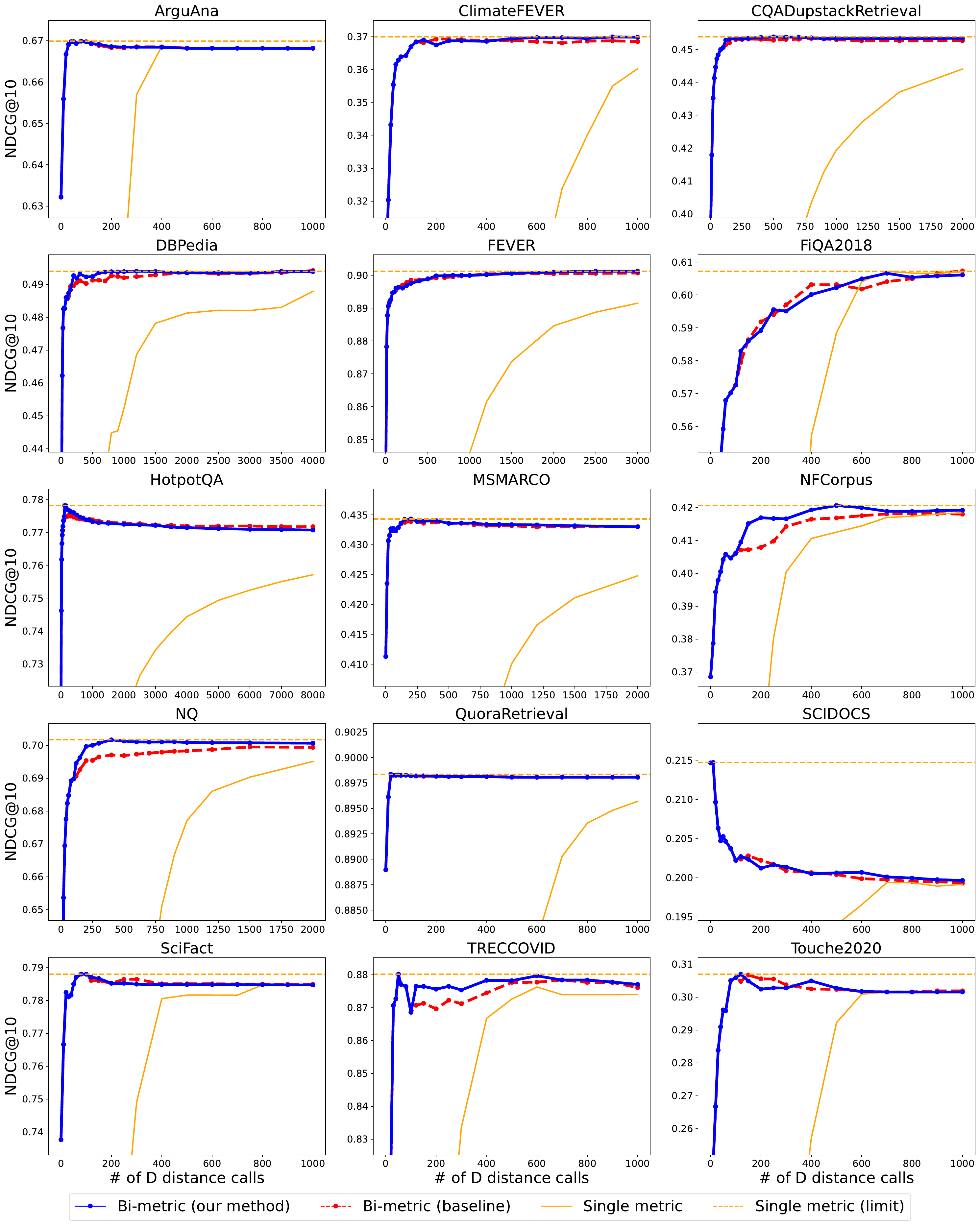}
\caption{Results for 15 MTEB Retrieval datasets. The x-axis is the number of expensive distance function calls. The y-axis is the NDCG@10 score. The cheap model is ``bge-base-en-v1.5'', the expensive model is ``SFR-Embedding-Mistral'', and the nearest neighbor search algorithm used is DiskANN.}
\label{fig:diskann-bge-base-ndcg}
\end{figure}

\begin{figure}[!h]
\centering
\includegraphics[width=0.99\textwidth]{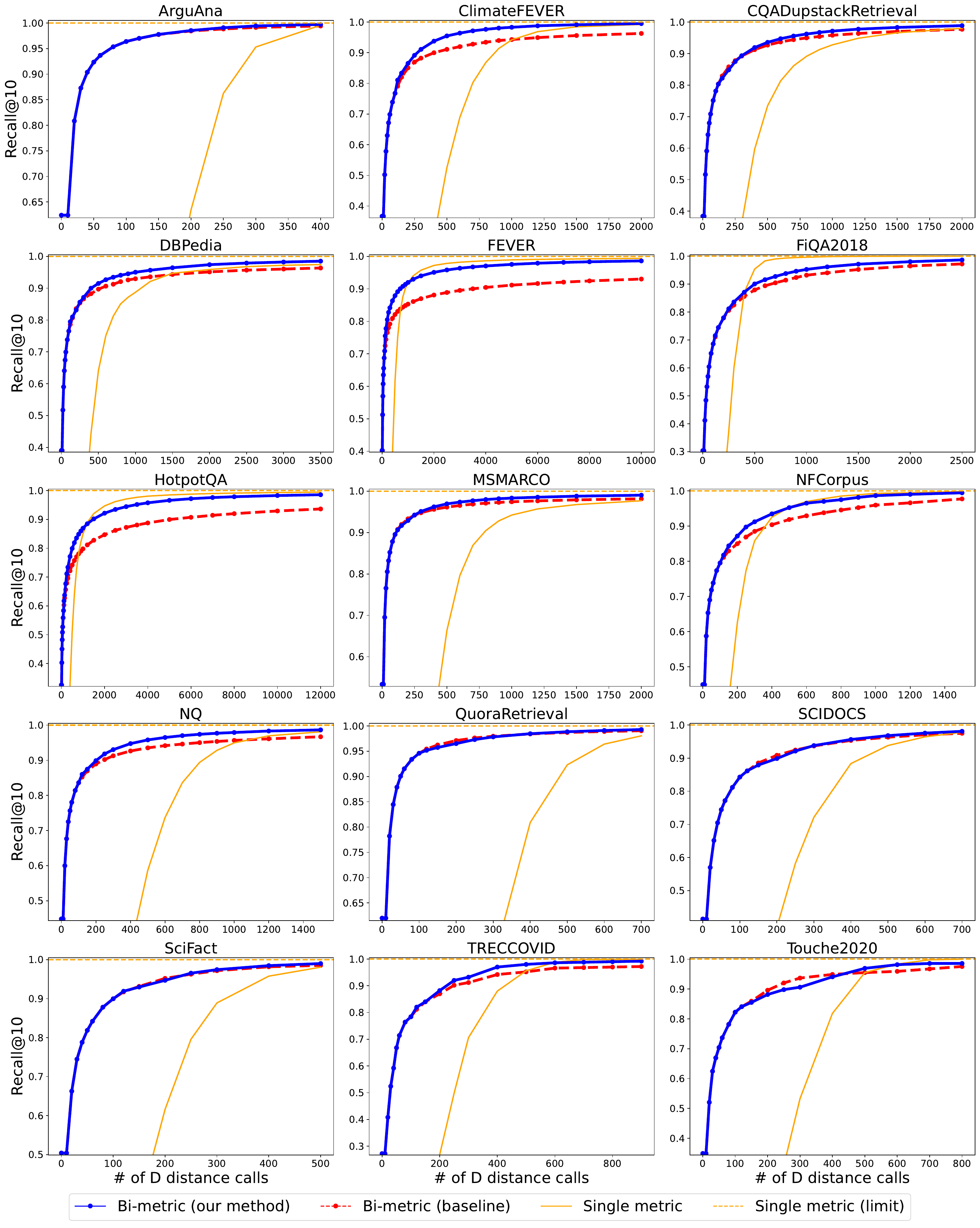}
\caption{Results for 15 MTEB Retrieval datasets. The x-axis is the number of expensive distance function calls. The y-axis is the Recall@10 score. The cheap model is ``bge-base-en-v1.5'', the expensive model is ``SFR-Embedding-Mistral'', and the nearest neighbor search algorithm used is DiskANN.}
\label{fig:diskann-bge-base-recall}
\end{figure}

\begin{figure}[!h]
\centering
\includegraphics[width=0.99\textwidth]{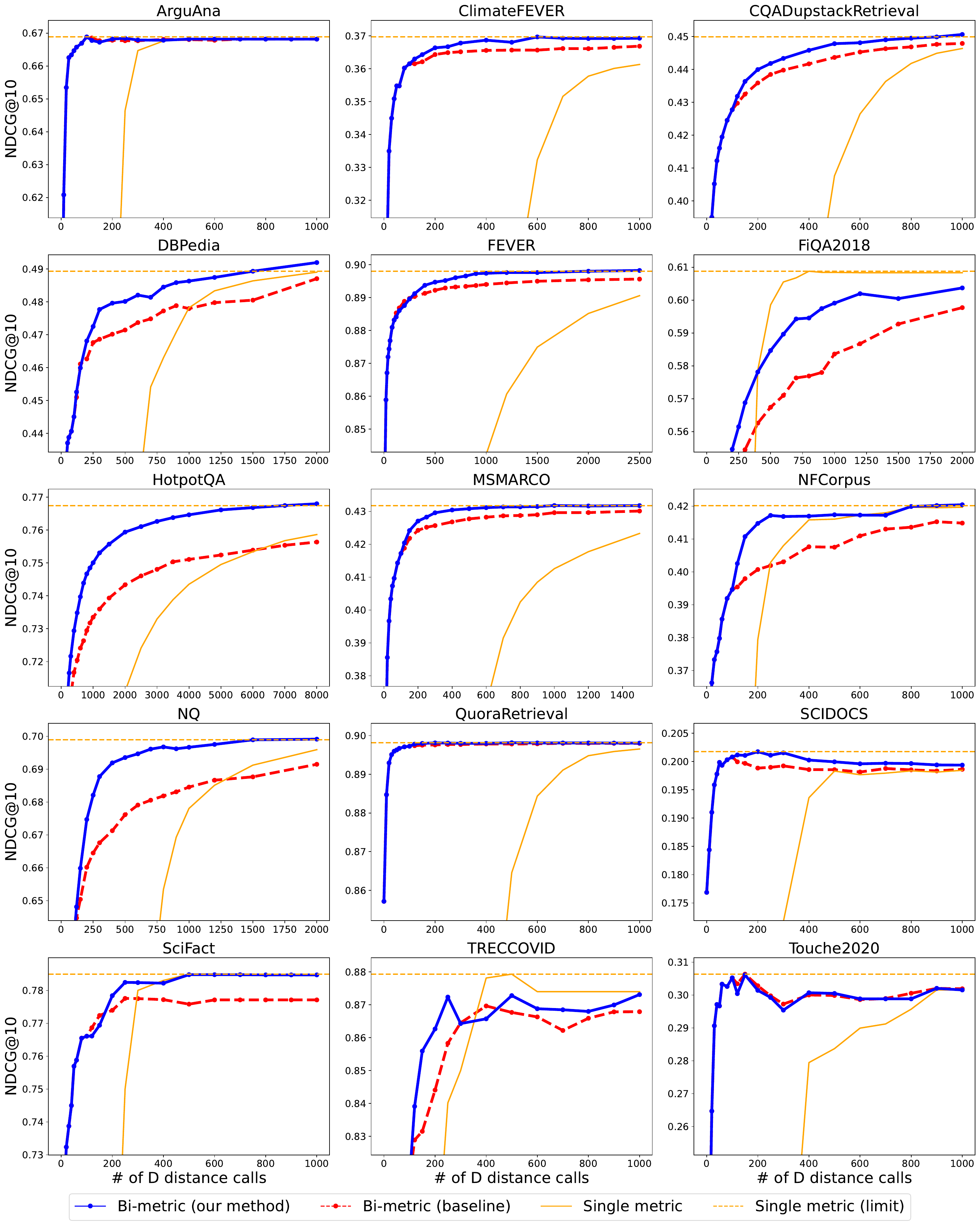}
\caption{Results for 15 MTEB Retrieval datasets. The x-axis is the number of expensive distance function calls. The y-axis is the NDCG@10 score. The cheap model is ``bge-micro-v2'', the expensive model is ``SFR-Embedding-Mistral'', and the nearest neighbor search algorithm used is NSG.}
\label{fig:nsg-bge-micro-ndcg}
\end{figure}

\begin{figure}[!h]
\centering
\includegraphics[width=0.99\textwidth]{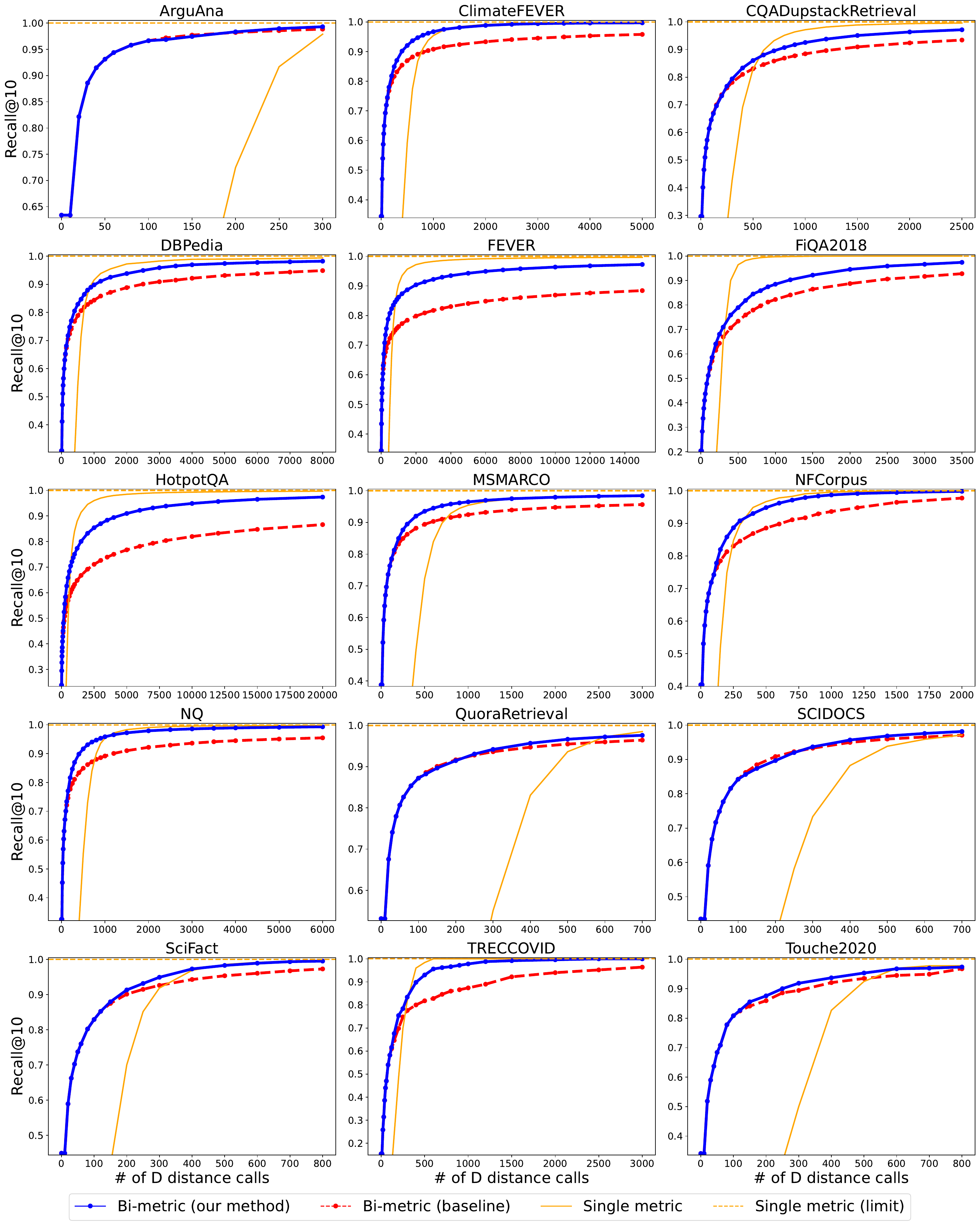}
\caption{Results for 15 MTEB Retrieval datasets. The x-axis is the number of expensive distance function calls. The y-axis is the Recall@10 score. The cheap model is ``bge-micro-v2'', the expensive model is ``SFR-Embedding-Mistral'', and the nearest neighbor search algorithm used is NSG.}
\label{fig:nsg-bge-micro-recall}
\end{figure}

\FloatBarrier

\end{document}